\documentclass[letterpaper, 10 pt, conference]{ieeeconf}  

\IEEEoverridecommandlockouts                              

\usepackage{graphicx}      
\usepackage{amsmath,amssymb,amsfonts}
\usepackage{epsfig} 
\usepackage{arydshln}
\usepackage{etoolbox}
\usepackage{ mathrsfs }
\usepackage{caption}
\usepackage{subcaption}
\usepackage{gensymb}
\usepackage{color}
\usepackage{mathtools}
\usepackage{systeme}
\usepackage{empheq}
\usepackage{float}
\usepackage{caption}
\usepackage{subcaption}
\usepackage{ntheorem}
\usepackage[outdir=./]{epstopdf}

\newtheorem{problem}{Problem}
\newtheorem{definition}{Definition}
\newtheorem{lemma}{Lemma}

\newtheorem{theorem}{Theorem}{}
{}
\newtheorem{remark}{Remark}{}

\begin{document}

\title{\LARGE \bf
Enforcing Safety under Actuator Attacks through Input Filtering
}

\author{C\'{e}dric Escudero$^{1}$, Carlos Murguia$^{2}$, Paolo Massioni$^{3}$, Eric Zama\"i$^{4}$
\thanks{$^{1}$C\'{e}dric Escudero is with Laboratoire Amp\`{e}re CNRS, INSA Lyon, Universit\'{e} de Lyon, 69621 Villeurbanne CEDEX, France,
        {\tt\small Email: cedric.escudero@insa-lyon.fr}}%
\thanks{$^{2}$Carlos Murguia is with the Department of Mechanical Engineering, Dynamics and Control Group, Eindhoven University of Technology, The Netherlands,
        {\tt\small Email: c.g.murguia@tue.nl}}%
\thanks{$^{3}$Paolo Massioni is with Laboratoire Amp\`{e}re CNRS, INSA Lyon, Universit\'{e} de Lyon, 69621 Villeurbanne CEDEX, France,
		{\tt\small Email: paolo.massioni@insa-lyon.fr}}%
\thanks{$^{1}$Eric Zama\"{i} is with Laboratoire Amp\`{e}re CNRS, INSA Lyon, Universit\'{e} de Lyon, 69621 Villeurbanne CEDEX, France,
		{\tt\small Email: eric.zamai@insa-lyon.fr}}%
}

\maketitle
\thispagestyle{empty}
\pagestyle{empty}

\begin{abstract}
Actuator injection attacks pose real threats to all industrial plants controlled through communication networks. In this manuscript, we study the possibility of constraining the controller output (i.e. the input to the actuators) by means of a dynamic filter designed to prevent reachability of dangerous plant states -- preventing thus attacks from inducing dangerous states by tampering with the control signals. The filter synthesis is posed as the solution of a convex program (convex cost with Linear Matrix Inequalities constraints) where we aim at shifting the reachable set of control signals to avoid dangerous states while changing the controller dynamics as little as possible. We model the difference between original control signals and filtered ones in terms of the H-infinity norm of their difference, and add this norm as a constraint to the synthesis problem via the bounded-real lemma. Results are illustrated through simulation experiments.
\end{abstract}

\section{Introduction}
Integration of information and communication technologies is rapidly increasing in many industrial applications and critical infrastructure. This integration leads to engineered systems being controlled by embedded computer devices over communication networks \cite{Lee_CPS} -- the so-called Network Control Systems NCSs. However, next to the advantages that the use of NCSs might offer, they possess increased vulnerabilities against adversarial attacks at the cyber-layer (software, computing hardware, and communications). NCSs are often opened to the outside world  (e.g., for remote control and maintenance via the internet or cellular 4G/5G networks) \cite{Laughlin}, and adversaries exploit the cyber-layer to launch cyberattacks, even to industrial NCSs \cite{fromics}. These security challenges have emerged as new research objectives for the control engineering community \cite{DIBAJI2019394}. In particular, \textit{deception} attacks  (e.g., spoofing, false-data injection, and replay attacks \cite{Teixeira}) have attracted considerable attention. These attacks tamper with systems' signals (sensing and control) to degrade their performance.

The secure control literature has come up with different methods to prevent deception attacks. Such methods protect the plant by redesigning controllers that mitigate the degradation of the plant induced by attacks \cite{Giraldo}. Among the prevention measures, set-theoretic methods \cite{BlanchiniBook} have already shown their potential to enforce that the plant states avoid dangerous states -- a subset of the state space that, if reached, compromises the system integrity and leads to system degradation. Set-theoretic methods range from the synthesis of secure controllers \cite{MURGUIA2020108757} in terms of invariant ellipsoids to the design of artificial controller saturation  \cite{Kafash,Kafashb,Kafash_des}. Such methods aim to prevent any class of attacks injecting signals into the communication network to damage the system. Other work \cite{Sinopoli_reach,Soodeh_rch,8712588} uses set-theoretic methods to model stealthy attacks and design fault detectors to mitigate the degradation of the plant \cite{Sinopoli_reach,8712588}.

In this manuscript, we focus on \textit{actuator injection attacks}, a class of deception attacks that inject malicious control signals into the plant. We assume the adversary is capable of injecting signals to true control actions by compromising either the controller, or the communication network that transmits the control actions to the actuators. We are interested in attacks that aim to damage the integrity of the system while they do not change the ``normal'' behaviour of the plant (e.g.\ reference tracking or stabilization to a set). We refer to this class of attacks as \textit{stealthy actuator attacks}. Previous results are reported in \cite{esc_iet}, which proposes preventing such attacks by limiting the controller output (the input to the actuators). The main drawback of this work is that it leads to strong limitations on the controller output as its modulus has to be smaller than a predefined threshold at all times. The latter can sometimes be too restrictive and prevent the achievement of the control objective. Here we propose an improvement to \cite{esc_iet} by allowing for dynamic thresholds. We achieve this by filtering the controller output before it reaches the actuators, and seeking for the filter dynamics that provides safety guarantees in the sense of avoiding dangerous states. Using this filter let us fine tune the limitations we want to enforce on the controller output, which leads to less conservative results compared to the static threshold considered in \cite{esc_iet} -- as now we have the freedom to impose constrains in the frequency domain. 

The rest of this manuscript is organized as follows. Section~\ref{sec:problem} formulates the problem we seek to address. Section~\ref{sec:tools} provides necessary mathematical tools (reachability and ellipsoidal approximations of reachable sets) to perform the filter synthesis, which is given in Section~\ref{sec:result}. The synthesis procedure is written in terms of a series of convex programs subject to Linear Matrix Inequalities (LMIs) constraints. Section~\ref{sec:example} illustrates the performance of our tools by simulations experiments.

\noindent
\emph{\textbf{Notation:}} The symbol $\mathbb{R}$ stands for the real numbers, $\mathbb{R}^{n \times m}$ is the set of real ${n \times m}$ matrices, and $\mathbb{R}_{>0}$ ($\mathbb{R}_{\geq 0}$) denotes the set of positive (non-negative) real numbers. Matrix $A^\top$ indicates the transpose of matrix $A$ and diag($a_1,...,a_n$) corresponds to a diagonal matrix with diagonal elements $a_1,...,a_n$.  The identity matrix of dimension $n$ is denoted by $I_n$, and $\mathbf{0}$ is a matrix of only zeros of appropriate dimensions. The notation $A \succeq 0$ (resp. $A \preceq 0$) indicates that the matrix $A$ is positive (resp. negative) semidefinite, i.e., all the eigenvalues of the symmetric matrix $A$ are positive (resp. negative) or equal to zero, whereas the notation \(A \succ 0\) (resp. \(A \prec 0\)) indicates the positive (resp. negative) definiteness, i.e., all the eigenvalues are strictly positive (resp. negative). The notation $\mathcal{E}_x(Q)$ stands for an ellipsoid of dimension $n$ with shape matrix $Q \in \mathbb{R}^{n \times n},\, Q = Q^\top  \succ 0$ and centered at zero, i.e.,
 $\mathcal{E}_x(Q):=\{x \in \mathbb{R}^n\, |\,x^\top Q x  \leqslant 1  \}$.

\section{PROBLEM FORMULATION} \label{sec:problem}
In this section, we introduce the class of systems and attacks under study, the problem formulation, and the proposed solution based on input filtering.

\begin{figure*}[t!]
\centering
\includegraphics[width=0.7\textwidth]{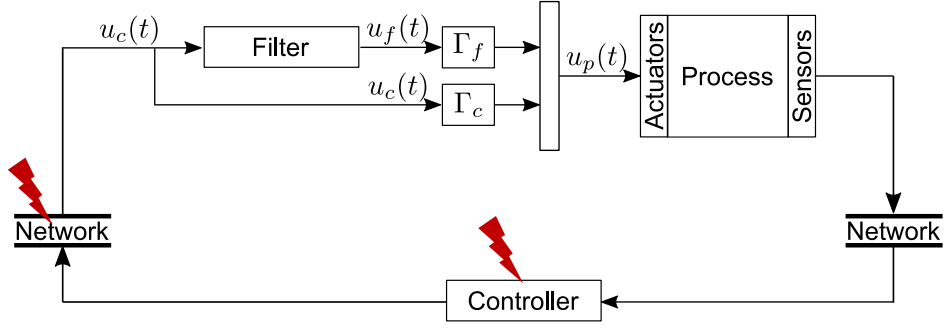}
\caption{Control system with the input filtering to prevent actuator injection attacks}
\label{fig:sys_with_filter}
\end{figure*}

\subsection{System Dynamics}
We consider linear time-invariant systems of the form:
\begin{equation} \label{eq:system}
\dot{x}_p(t) = A_p x_p(t) + B_p u_p(t),
\end{equation}
with time $t \in \mathbb{R}_{>0}$, state $x_p(t) \in \mathbb{R}^{n_p}$, control input $u_p(t) \in \mathbb{R}^{m}$, system matrices $A_p \in  \mathbb{R}^{n_p \times n_p}$ and $B_p  \in \mathbb{R}^{n_p \times m}$, and controllable pair $(A_p,B_p)$. Matrix $A_p$ is Hurwitz, i.e., the origin of \eqref{eq:system}, with $u_p(t)=\mathbf{0}$, $t \in \mathbb{R}$, is globally asymptotically stable.

The system description in \eqref{eq:system} comprises the actuators and process dynamics (that is the plant), i.e., some states are due to actuators and some other are due to physical variables in the system. The system is assumed to be part of a control-loop as illustrated in Figure~\ref{fig:sys_with_filter}, where \emph{the control block receives some of the states $x_p(t)$ (the measured states) to compute control actions $u_c(t)$, which are sent back to the actuators/process through public/unsecured communication networks}.

\begin{remark}
In the standard configuration, we have $u_c(t)=u_p(t)$, i.e, signals sent by the controller $u_c(t)$ equal applied control inputs $u_p(t)$. We make a distinction here because when we place the proposed safety-preserving filter in the loop, $u_c(t)$ and $u_p(t)$ will be in general different (as $u_p(t)$ will depend on the filter dynamics), see Figure~\ref{fig:sys_with_filter}.
\end{remark}

Control actions coming from the control block, $u_c(t)$, are peak-bounded inside a \emph{known} ellipsoid $\mathcal{E}_u(R,\bar{u})$ representing amplitude limitations of actuation signals (physical or imposed by design), i.e., control inputs $u_c(t)$ belong to the ellipsoidal set
\begin{equation} \label{eq:cons_input}
\mathcal{E}_u(R,\bar{u}) := \left\{ u \in \mathbb{R}^m |(u(t) - \bar{u})^\top R (u(t) -\bar{u}) \leqslant 1 \right\},
\end{equation}
for some known positive definite matrix $R \in \mathbb{R}^{m\times m}$ and vector $\bar{u} \in \mathbb{R}^{m}$.

\subsection{Adversarial Capabilities}

In this manuscript, we focus on \textit{False Data Injection attacks to actuators}. That is, we assume the adversary is capable of injecting signals to true control actions, $u_c(t)$, by compromising either the controller block itself (for instance by hacking into the processor) or the communication network that transmits $u_c(t)$ to actuators (see Figure~\ref{fig:sys_with_filter}). We consider two types of False Data Injection attacks to actuators, \emph{stealthy and non-stealthy}. Stealthy actuator attacks aim to damage the integrity of the system while letting the process states to operate normally. That is, they are attacks that do not change the ``normal" behaviour of the process (for instance, reference tracking or stabilization to a set) but drive the system dynamics to a part of the state space where physical degradation occurs -- e.g., car collisions, pipes breaking, accelerated wear and tear, explosions in power generators, etc. Non-stealthy attacks aim to induce fast and/or large damage regardless of their chances of being detected (thus they are not constrained by the normal operation set). Before we give a formal definition of these attacks, we introduce the notions of normal operation sets and safe sets.

\begin{definition}[Normal Operation Set] \label{def:attack_Type_N}
The normal operation set $X_n \subseteq \mathbb{R}^{n_p}$ for system \eqref{eq:system} is the set of states $x_p \in X_n$ where trajectories of the \emph{attack-free} system dynamics are expected to be contained in.
\end{definition}

So, if the states trajectories $x_p(t)$ remain inside the set $X_n$, the system is expected to be operating as usual, and thus no suspicion of attacks can arise. Hence, if adversaries aim to be stealthy, the trajectories of the attacked system must remain inside $X_n$. Stealthiness heavily constrains what the adversary can induce in the system as she/he is restricted to trajectories that are somehow standard in the process (and thus nondestructive). It follows that fast and/or large degradation might be hard to accomplish via stealthy attacks. On the other hand, non-stealthy attacks are easily spotted, which would make it easier for operators to run counter measures. There is a trade-off here, stealthy attacks lead to smaller but persistent degradation, and non-stealthy to larger/faster damage but short-lived attacks. Here we cover both stealthy and non-stealthy attacks. Results vary slightly from one case to the other. The difference mainly lies in the use of the normal operation set $X_n$ when we consider stealthy attacks.

We now need a degradation metric to make sense of the system safety level in the presence of actuator attacks. To this end, we introduce the following notion of \emph{safe sets}.

\begin{definition}[Safe Set] \label{def:attack_Type_safe}
The safe set $X_s \subseteq \mathbb{R}^n$ for system \eqref{eq:system} is the set of states $x_p \in X_s$ where the safe and proper operation of the system is guaranteed. The safe set $X_s$ is the part of the state space that excludes critical states -- states that, if reached, compromise the system physical integrity.
\end{definition}

Critical states might represent states in which, for instance, the pressure of a holding vessel exceeds its pressure rating, negative inter-vehicle distances lead to collisions in cooperative driving, or the level of a liquid in a tank exceeds its capacity. Safe sets exclude, by definition, all critical states from the state space of \eqref{eq:system}.

\begin{definition}[Actuator (Stealthy) Attacks] \label{def:attack_Type_S}
Attacks that tamper with control inputs by injecting signals to true control actions, $u_c(t)$, and aim to degrade the operation of the system dynamics by pushing trajectories outside the safe set $X_s$ (while keeping them inside the normal operation set $X_n$ for stealthiness).
\end{definition}

\subsection{Safety-Preserving Filters}
We propose to protect the plant against actuator attacks by filtering control actions, $u_c(t)$, before they reach the actuators. That is, we pass $u_c(t)$ through a filter to enforce, by design, that it is impossible for actuator attacks to drive the system outside the safe set. The filter output, $u_f(t)$, is the filtered control signal that is sent to the actuators, see Figure~\ref{fig:sys_with_filter}. We consider linear time-invariant filters of the form:
\begin{equation} \label{eq:filter}
\begin{gathered}
\dot{x}_f(t) = A_fx_f(t) + B_fu_c(t),\\
u_f(t) = C_f x_f(t) + D_f u_c(t),
\end{gathered}
\end{equation}
with filter state $x_f(t) \in \mathbb{R}^{n_f}$, filter input $u_c(t) \in \mathbb{R}^{m}$ (the original control signal sent by the control block), filter output $u_f(t) \in \mathbb{R}^{m}$ to be transmitted to the plant, and  matrices $A_f \in \mathbb{R}^{n_f \times n_f}$, $B_f \in \mathbb{R}^{n_f \times m}$, $C_f \in \mathbb{R}^{m \times n_f}$, and $D_f \in \mathbb{R}^{m \times m}$ to be designed. We allow for partial filtering in the sense that not all filtered inputs $u_f$ reach the plant. We allow for some $u_c$ to get through the filter and reach it directly. It follows that the control signal driving the plant, $u_p$, can be written as follows:
\begin{equation}
u_p(t) = \Gamma_c u_c(t) + \Gamma_f u_f(t),
\end{equation}
where $\Gamma_c \in \mathbb{R}^{m \times m}$ and $\Gamma_f \in \mathbb{R}^{m \times m}$ are diagonal matrices used for selecting which control signals are unfiltered and filtered, respectively. Matrices $\Gamma_c$ and $\Gamma_f$ satisfy
\begin{equation}
\Gamma_c + \Gamma_f = I_m.
\end{equation}
Define the extended state $\zeta := [x_p^\top, x_f^\top]^\top \, \in \mathbb{R}^n$, $n = n_p + n_f$. Then, the filter and plant can be stacked together as:
\begin{equation} \label{eq:extended_system}
\dot{\zeta}(t) = A \zeta(t) + B u_c(t),
\end{equation}
with
\begin{equation} \label{eq:matrice}
A := \begin{bmatrix} A_p & B_p \Gamma_f C_f\\
					\mathbf{0} & A_f\end{bmatrix},
								  	 \quad B :=
								  	\begin{bmatrix}B_p \Gamma_f D_f + B_p \Gamma_c\\ B_f\end{bmatrix},						  	
\end{equation}
$A \in \mathbb{R}^{n \times n}$, and $B \in \mathbb{R}^{n \times m}$.

Now we can state the problem we seek to address.

\begin{problem}\label{problem1a}
Given the system dynamics \eqref{eq:system}, the filter \eqref{eq:filter}, the safe-set $X_s$ in Definition \ref{def:attack_Type_safe}, and the normal operation set $X_n$ in Definition \ref{def:attack_Type_N}, find (if possible) filter matrices ($A_f$,$B_f$,$C_f$,$D_f$) such that the system asymptotic trajectories are contained in $X_s$ for all actuator attacks satisfying Definition \ref{def:attack_Type_S}
\end{problem}

The solution to Problem 1 aims to enforce that the steady state trajectories of \eqref{eq:system}, in series interconnection with the filter \eqref{eq:filter}, are constrained inside the safe set $X_s$.

\section{Preliminary Tool} \label{sec:tools}


We first introduce the reachable set of the extended system \eqref{eq:extended_system} as we will work on this set to enforce safety.

\begin{definition}[Reachable Set \cite{Gayek}] \label{def:reachset2}
The reachable set $\mathcal{R}_{\zeta}(t)$ at time $t \in \mathbb{R}_{>0}$ from initial condition $\zeta(t_0) \in \mathbb{R}^{n}$ is the set of extended states $\zeta(t)$ that satisfy the extended differential equations \eqref{eq:extended_system}, over all control actions $u_c(t)$ satisfying $u_c(t) \in \mathcal{E}_u(R,\bar{u})$, i.e.,
\begin{equation}\label{reachset1}
\mathcal{R}_{\zeta}(t):= \left\{ \zeta(t) \left|
\begin{split}
&\zeta(t_0) \in \mathbb{R}^{n},\\
&\zeta(t) \hspace{1mm}\text{satisfies \eqref{eq:extended_system}},\\
&\text{and }u_c(t) \in \mathcal{E}_u(R,\bar{u}).
\end{split}
\right.
\right\}.
\end{equation}
\end{definition}
We denote the asymptotic reachable set (the ultimate bound on $\mathcal{R}_{\zeta}(t)$) as $\mathcal{R}_{\zeta}(\infty) := \lim_{t \rightarrow \infty }\mathcal{R}_{\zeta}(t)$. Note that, because the input $u_c$ is bounded, the asymptotic set $\mathcal{R}_{\zeta}(\infty)$ always exists if $A$ in \eqref{eq:extended_system} is Hurwitz (which is true when the filter matrix $A_f$ and the plant matrix $A_p$ are both Hurwitz because of the block triangular structure of $A$).

\subsection{Ellipsoidal Bound on $\mathcal{R}_{\zeta}(\infty)$}
The set-theoretic method proposed in this manuscript for synthesizing the filter is based on outer approximations of the asymptotic reachable set $\mathcal{R}_{\zeta}(\infty)$ of \eqref{eq:extended_system}. Because the exact computation of $\mathcal{R}_{\zeta}(\infty)$ is not tractable, the proposed method relies on an outer ellipsoidal approximation $\mathcal{E}_\zeta(Q)$ of $\mathcal{R}_{\zeta}(\infty)$, i.e., $\mathcal{R}_{\zeta}(\infty) \subseteq \mathcal{E}_\zeta(Q)$ (referred hereafter as an ellipsoidal bound on $\mathcal{R}_{\zeta}(\infty)$).

\begin{definition} \label{def_invariant set}
The ellipsoidal set $\mathcal{E}_\zeta(Q)$ is invariant for the dynamical system \eqref{eq:extended_system}, if for all initial states $\zeta(t_0) \in \mathcal{E}_\zeta(Q)$, and all $u_c(t) \in \mathcal{E}_u(R,\bar{u})$, the trajectories $\zeta(t)$ of \eqref{eq:extended_system} satisfy $\zeta(t) \in \mathcal{E}_\zeta(Q), \forall \hspace{1mm} t \ge t_0$.
\end{definition}


\begin{remark}
Note that, by definition, any invariant ellipsoidal set $\mathcal{E}_\zeta(Q)$, in the sense of Definition \ref{def_invariant set}, contains the reachable set $\mathcal{R}_{\zeta}(\infty)$ in \eqref{reachset1}. Hence, any invariant set $\mathcal{E}_\zeta(Q)$ is an ellipsoidal bound on $\mathcal{R}_{\zeta}(\infty)$.
\end{remark}

In a previous work \cite{esc_iet}, we have provided sufficient conditions for ellipsoid sets to be invariant for a class of LTI systems. This method is based on the search of a Lyapunov-like function, $V(\zeta)=\zeta^\top Q \zeta$, using Linear Matrix Inequalities (LMIs) \cite{boyd1994linear}. Before recalling this result, we model the normal operation set, $X_n$, defined in Definition~\ref{def:attack_Type_N} as an ellipsoid $\mathcal{E}_n(\Xi,\bar{\xi})$ satisfying

\begin{equation} \label{eq:cons_output}
(\zeta(t) - \bar{\xi})^\top \Xi (\zeta(t) -\bar{\xi}) \leqslant 1,
\end{equation}
with
\begin{equation} \label{eq:matricess}
\Xi = \begin{bmatrix} \Xi_p & \mathbf{0}\\
					\mathbf{0} & \mathbf{0}\end{bmatrix}
\,
,
\,
\bar{\xi} = \begin{bmatrix} \bar{\xi}_p\\ \mathbf{0}\end{bmatrix},						  	
\end{equation}
for some known positive semi-definite matrix $\Xi_p \in \mathbb{R}^{n_p \times n_p}$ and vector $\bar{\xi}_p \in \mathbb{R}^{n_p}$. Note that $\Xi_p$ is in general rank-deficient, as it only constrains some of the plant states $x_p(t)$ --  $\mathcal{E}_n(\Xi,\xi)$ can even coincide with $\mathbb{R}^{n\times n}$ by picking $\Xi_p = 0$.

Next, we state the preliminary tool used to find invariant ellipsoidal sets for the extended dynamics \eqref{eq:extended_system}. Once we have found such a set, we use it as an ellipsoidal bound to the reachable set $\mathcal{R}_{\zeta}(\infty)$ in \eqref{reachset1} (Remark 2).

\begin{lemma}[Invariant Ellipsoidal Set]
Consider system \eqref{eq:extended_system} with system matrices as defined in \eqref{eq:matrice}. If there exist matrix $Q \in \mathbb{R}^{n \times n}$ and constants $\alpha,\beta,\lambda \in \mathbb{R}_{\geq 0}$ satisfying the following inequalities:

\begin{equation} \label{eq:thun1}
-E - \alpha F - \beta S - \lambda T \succeq 0,
\end{equation}
\begin{equation} \label{eq:thun2}
Q  \succ 0,
\end{equation}
with
\begin{align}
E = &\begin{bmatrix} \label{eq:lem_M}
A^\top Q + Q A\,  & \mathbf{0}\, & Q B\\
\ast\, &  \mathbf{0}\, & \mathbf{0}\\
\ast\, &  \ast\, & \mathbf{0}
\end{bmatrix},\\
F = &\begin{bmatrix} \label{eq:lem_N}
Q\, & \mathbf{0}\,  & \mathbf{0}\\
\ast\, & -1\, & \mathbf{0}\\
\ast\, & \ast\, & \mathbf{0}
\end{bmatrix},\\
S = &\begin{bmatrix}
\mathbf{0}\, & \mathbf{0}\, & \mathbf{0}\\
\ast\, & 1-\bar{u}^\top R \bar{u}\, & \bar{u}^\top R\\
\ast\, & \ast\, & -R
\end{bmatrix},\\
T = &\begin{bmatrix}
-\Xi\, & \Xi \bar{\xi}\, & \mathbf{0}\\
\ast\, & 1-\bar{\xi}^\top \Xi \bar{\xi}\, & \mathbf{0}\\
\ast\, & \ast\, & \mathbf{0}
\end{bmatrix};
\end{align}
\noindent then,
$\zeta(t_0) ^\top Q \zeta(t_0) \leqslant 1 \Rightarrow
\zeta(t) ^\top Q \zeta(t) \leqslant 1$, for all $t \ge t_0$, $u_c(t) \in \mathcal{E}_u(R, \bar{u})$, and $\zeta(t) \in \mathcal{E}_n(\Xi, \bar{\xi})$.
\label{lem:invar}
\end{lemma}


\section{Solution to Problem 1} \label{sec:result}

In this section, we propose a synthesis framework, built around Lemma~\ref{lem:invar}, to find filter matrices solving Problem 1 in terms of the solution of a series of semidefinite programs. 

\noindent To prevent damage from actuator injection attacks, the plant states $x_p(t)$ need to remain inside the safe set $X_s$. We model $X_s$ as an ellipsoid $\mathcal{E}_s(\Psi, \bar{\psi})$ with positive semi-definite $\Psi \in \mathbb{R}^{n_p \times n_p}$ and vector $\bar{\psi} \in \mathbb{R}^{n_p}$. Matrix $\Psi$ is in general rank-deficient, as only one part of the physical state might be subject to the safe zone. Hence, to enforce safety, we want to guarantee the following

\begin{equation}
x_p(t) \in \mathcal{E}_{x_p}(Q_{x_p}) \Rightarrow x_p(t) \in \mathcal{E}_s(\Psi, \bar{\psi}),\,\,\, 
\end{equation}
where $\mathcal{E}_{x_p}(Q_{x_p})$ is the projection of $\mathcal{E}_\zeta(Q)$ onto the $x_p$-hyperplane -- because we are only interested in safety of the plant states, not the filter states. Matrix $Q_{x_p} \in \mathbb{R}^{n_p \times n_p}$ can be written in terms of matrix $Q$ of $\mathcal{E}_\zeta(Q)$ as  $Q_{x_p} = Q_1 - Q_2 Q_3^{-1} Q_2^\top$, where  \begin{equation} \label{eq:proj}
\begin{bmatrix}
Q_1\, & Q_2\\
\ast\, & Q_3
\end{bmatrix} := Q,
\end{equation}
with $Q_1 \in \mathbb{R}^{n_p \times n_p}$, $Q_2 \in \mathbb{R}^{n_p \times n_f}$, $Q_3 \in \mathbb{R}^{n_f \times n_f}$, see \cite{MURGUIA2020108757} for details.

Note that, by filtering $u_c$, we are changing the dynamics of control signals. We do not want to make $u_c$ and the filtered $u_f$ overly different. To this end, we introduce a distortion constraint in terms of the $H_\infty$-norm between the original, $u_c$, and filtered $u_f$. Define the error $z(t) := u_f(t) - u_c(t)$. It is easy to verify that $z(t)$ can be written in terms of the extended state $\zeta(t)$ and $u_c(t)$, as $z(t) = C_z \zeta(t) + D_z u_c(t)$ with matrices:
\begin{equation}
C_z :=
\begin{bmatrix}
\mathbf{0}\, & C_f
\end{bmatrix}
\,
,
\,
D_z :=
\begin{bmatrix}
D_f - I_m
\end{bmatrix}.
\end{equation}
We treat this $z(t)$ as a performance output for the extended dynamics \eqref{eq:extended_system}. For system \eqref{eq:extended_system}, with input $u_c(t)$ and output $z(t)$, let $T_{u_c \rightarrow z}(s)$ denote the transfer matrix from $u_c(t)$ to $z(t)$, i.e., $T_{u_c \rightarrow z}(s) := C_z(sI_n - A)^{-1}B + D_z$. We use the $H_\infty$-norm of $T_{u_c \rightarrow z}(s)$ as a metric to quantify how different $u_c$ and $u_f$ are. If no filter is in place, this norm is trivially zero, and as we let them be more different, the norm will grow unbounded. When designing the filter to enforce safety, we also want to keep the $H_\infty$-norm of $T_{u_c \rightarrow z}(s)$ below a predefined level $\gamma \in \mathbb{R}_{\geq 0}$, i.e. $\lvert \lvert T_{u_c \rightarrow z}(s) \rvert \rvert_{H_\infty} \le \gamma$. We use this gamma to modulate how much we are willing to sacrifice in terms of control performance to enforce safety. 

We have all the ingredients now to re-cast Problem 1 above in terms of our new notation. 

\begin{problem}[Filter Synthesis Problem]
Find the filter matrices $\kappa :=$($A_f$, $B_f$, $C_f$, $D_f$) such that (i) the ellipsoid $\mathcal{E}_\zeta(Q)$ is invariant for the dynamical system in \eqref{eq:extended_system} with $u_c(t) \in \mathcal{E}_u(R,\bar{u})$ and $\zeta(t) \in \mathcal{E}_n(\Xi,\bar{\xi})$, (ii) the projection of the invariant ellipsoid $\mathcal{E}_\zeta(Q)$ onto the $x_p$-hyperplane, i.e. $\mathcal{E}_{x_p}(Q_{x_p})$ is a subset of the safe set $\mathcal{E}_s(\Psi,\bar{\psi})$, and (iii) the $H_\infty$-norm of $T_{u_c \rightarrow z}(s)$ is upper bounded by $\gamma$. 
\end{problem}

Finally, just before formulating our main result, we briefly describe the procedure for obtaining the main result. Because $\kappa:=$ ($A_f$, $B_f$, $C_f$, $D_f$) are variables in the synthesis problem, the blocks $Q A$ and $Q B$ in \eqref{eq:lem_M} are nonlinear in ($\kappa$,$Q$). Following the results in \cite{Scherer_lin}, we propose an invertible linearizing change of variables such that, in the new variables, the objective of the optimization problem \textbf{OP} is convex and the constraints are affine.

Let $Q$ be positive definite and of the form:
\begin{equation}
Q = \begin{bmatrix}
Y\, & N\\
N^\top\, & \tilde{Y}
\end{bmatrix}\, , \,
Q^{-1} = \begin{bmatrix}
X\, & M\\
M^\top\, & \tilde{X}
\end{bmatrix}
\end{equation}
where $Y$, $N$, $\tilde{Y}$, $X$, $M$, $\tilde{X}$ $\in \mathbb{R}^{n \times n}$; and $Y$, $\tilde{Y}$, $X$, $\tilde{X}$ are positive definite matrices. Define the following matrices
\begin{equation} \label{eq:Pi}
\Pi_1 := \begin{bmatrix}
X\, & I_n\\
M^\top\, & \mathbf{0}
\end{bmatrix}\, , \,
\Pi_2 := \begin{bmatrix}
I_n\, & Y\\
\mathbf{0}\, & N^\top
\end{bmatrix}.
\end{equation}
It is easy to verify that $Q \Pi_1 = \Pi_2$. Define the change of filter variables as follows:
\begin{equation} \label{eq:change}
    \begin{cases}
      \hat{A}_f := Y A_p X + Y B_p \Gamma_f C_f M^\top + N A_f M^\top\\
      \hat{B}_f := Y B_p \Gamma_f D_f + N B_f\\
      \hat{C}_f := C_f M^\top\\
      \hat{D}_f := D_f
    \end{cases}
\end{equation}
with $\hat{A}_f \in \mathbb{R}^{n_p \times n_p}$, $\hat{B}_f \in \mathbb{R}^{n_p \times m}$, $\hat{C}_f \in \mathbb{R}^{m \times n_p}$, $\hat{D}_f \in \mathbb{R}^{m \times m}$. Note that if $M$ and $N$ have full row rank, and $\hat{A}_f$, $\hat{B}_f$, $\hat{C}_f$, $\hat{D}_f$, $X$, and $Y$ are given, we can extract the true filter matrices $A_f$, $B_f$, $C_f$, $D_f$ satisfying \eqref{eq:change}. 

We can now formulate our main result.

\begin{theorem} \label{th:two}
Consider system \eqref{eq:extended_system} with system matrices as defined in \eqref{eq:matrice}. \noindent If $\Xi_p$ and $\Psi$ are invertible, and there exist $X \in \mathbb{R}^{n \times n}$, $Y \in \mathbb{R}^{n \times n}$ with $n_f = n_p$, $\hat{A}_f \in \mathbb{R}^{n_p \times n_p}$, $\hat{B}_f \in \mathbb{R}^{n_p \times m}$, $\hat{C}_f \in \mathbb{R}^{m \times n_p}$, $\hat{D}_f \in \mathbb{R}^{m \times m}$, and $\alpha$, $\beta$, $\lambda$, $\delta$, $\gamma$, $\epsilon$ $\in \mathbb{R}_{\geq 0}$ for which the following inequalities are satisfied:

\begin{equation} \label{eq:thTwo1}
-E' - \alpha F' - \beta S' - \lambda T' \succeq 0,
\end{equation}
\begin{equation} \label{eq:thTwo2}
-J - \delta W \succeq 0,
\end{equation}
\begin{equation} \label{eq:thTwo3}
-L \succeq 0,
\end{equation}
\begin{equation} \label{eq:thTwo4}
\mathcal{Q}(\kappa) \succ 0.
\end{equation}
with
\begin{align*}
E' = &\begin{bmatrix}
\mathcal{A}(\kappa)^\top + \mathcal{A}(\kappa)\,  & \mathbf{0}\, & \mathcal{B}(\kappa)\\
\ast\, &  \mathbf{0}\, & \mathbf{0}\\
\ast\, &  \ast\, & \mathbf{0}
\end{bmatrix},\\
F' = &\begin{bmatrix}
\mathcal{Q}(\kappa)\, & \mathbf{0}\,  & \mathbf{0}\\
\ast\, & -1\, & \mathbf{0}\\
\ast\, & \ast\, & \mathbf{0}
\end{bmatrix},\\
S' = &\begin{bmatrix}
\mathbf{0}\, & \mathbf{0}\, & \mathbf{0}\\
\ast\, & 1-\bar{u}^\top R \bar{u}\, & \bar{u}^\top R\\
\ast\, & \ast\, & -R
\end{bmatrix},\\
T' = &\begin{bmatrix}
-\mathcal{G}\, & \mathcal{H}\, & \mathbf{0}\\
\ast\, & 1-\bar{\xi}^\top \Xi \bar{\xi}\, & \mathbf{0}\\
\ast\, & \ast\, & \mathbf{0}
\end{bmatrix},\\
J = &\begin{bmatrix}
\mathbf{0}\, & -X \Psi \bar{\psi}\, & -X\\
\ast\, & -1+\bar{\psi}^\top \Psi \bar{\psi}\, & \mathbf{0}\\
\ast\, & \ast\, & -\Psi^{-1}
\end{bmatrix},\\
W = &\begin{bmatrix}
-X\, & \mathbf{0}\, & \mathbf{0}\\
\ast\, & 1\, & \mathbf{0}\\
\ast\, & \ast\, & \mathbf{0}
\end{bmatrix},\\
L = &\begin{bmatrix}
\mathcal{A}(\kappa)^\top + \mathcal{A}(\kappa)\, & \mathcal{B}(\kappa)\, & \mathcal{C}_z(\kappa)^\top\\
\ast\, & -(\gamma - \epsilon)I\, & D_z^\top\\
\ast\, & \ast\, & -\gamma I
\end{bmatrix},
\end{align*}
where
\begin{align*}
&\Pi_1^\top Q A \Pi_1 =
\begin{bmatrix}A_p X + B_p \Gamma_f \hat{C_f}\, & A_p\\
			\hat{A}_f\, & Y A_p
\end{bmatrix} =: \mathcal{A}(\kappa)
\,
,
\\
&\Pi_1^\top Q B =
\begin{bmatrix} B_p \Gamma_f \hat{D}_f + B_p \Gamma_c\\
			 \hat{B}_f + Y B_p \Gamma_c
\end{bmatrix} =: \mathcal{B}(\kappa)
\,
,
\\
&C_z \Pi_1 =
\begin{bmatrix} \hat{C}_f\, & \mathbf{0}
\end{bmatrix} =: \mathcal{C}_z(\kappa)
\,
,
\\
&\Pi_1^\top Q \Pi_1 =
\begin{bmatrix} X\, & I\\
			 I\, & Y
\end{bmatrix} =: \mathcal{Q}(\kappa)
\,
,
\\
&\Pi_1^\top \Xi \bar{\xi} =
\begin{bmatrix} X \Xi_p \bar{\xi}_p\\
			\Xi_p \bar{\xi}_p
\end{bmatrix} =: \mathcal{H}
\,
,
\\
&\begin{bmatrix} 2X - \Xi_p^{-1}\, &  X \Xi_p\\
			 \Xi_p X\, &  \Xi_p
\end{bmatrix} =: \mathcal{G};
\end{align*}
\noindent then,
$\zeta(t_0) ^\top Q \zeta(t_0) \leqslant 1 \Rightarrow
\zeta(t) ^\top Q \zeta(t) \leqslant 1$, $\mathcal{E}_{x_p}(X^{-1}) \subseteq \mathcal{E}_s(\Psi, \bar{\psi})$, and $\lvert \lvert T_{u_c \rightarrow z}(s) \rvert \rvert_{H_\infty} \le \gamma$, for all $t \ge t_0$, $u_c(t) \in \mathcal{E}_u(R, \bar{u})$, and $\zeta(t) \in \mathcal{E}_n(\Xi, \bar{\xi})$.
\end{theorem}
\begin{proof}
Firstly, consider \eqref{eq:thTwo3} then the Schur complement of the lower-right corner block matrix $-\gamma I$ of the matrix $L$ is the matrix defined by
\begin{align}
L &= \nonumber \\
&\begin{bmatrix}
\mathcal{A}(\kappa)^\top + \mathcal{A}(\kappa)\, & \mathcal{B}(\kappa)\\
\ast\, & -(\gamma - \epsilon)I
\end{bmatrix}
+
\frac{1}{\gamma}
\begin{bmatrix}
\mathcal{C}_z(\kappa)^\top\\
D_z^\top
\end{bmatrix}
\begin{bmatrix}
\mathcal{C}_z(\kappa)^\top\\
D_z^\top
\end{bmatrix}^\top
\end{align}
Then, left and right multiply by $[\zeta(t)^\top, u_c(t)^\top]^\top$; this implies that $\dot{V}(\zeta) -(\gamma - \epsilon)\lvert \lvert u_c(t) \rvert \rvert^2 + \frac{1}{\gamma} \lvert \lvert z(t) \rvert \rvert^2 \le 0$, which implies that $\lvert \lvert T_{u_c \rightarrow z}(s) \rvert \rvert_{H_\infty} \le \gamma$, under the change of variables.

Secondly, consider \eqref{eq:thTwo2}; left and right multiply by $[\zeta(t)^\top, 1, u_c(t)^\top]^\top$, and consider $\delta W$ as an S-procedure term by a positive multiplier $\delta$; this implies with the S-procedure under the change of variables:

\begin{align*}
&-[\zeta^\top, 1, u_c^\top] \,J\, [\zeta^\top, 1, u_c^\top]^\top \ge 0 \Leftrightarrow x_p(t) \in \mathcal{E}_s(\Psi, \bar{\psi})\\
&\text{when}\\
&[\zeta^\top, 1, u_c^\top] \, W \, [\zeta^\top, 1, u_c^\top]^\top \ge 0 \Leftrightarrow x_p(t) \in \mathcal{E}_{x_p}(X^{-1})
\end{align*}
as the projection of $\mathcal{E}_\zeta$ onto the $x_p$-hyperplane is defined as $\mathcal{E}_{x_p}(Y - N \tilde{Y}^{-1} N^\top)$ from \eqref{eq:proj} with $Y - N \tilde{Y}^{-1} N^\top = X^{-1}$ using block matrix inversion formulas.

This means that the plant state trajectories never leave the safe set for any plant state trajectories inside the invariant ellipsoidal set.

Lastly, consider \eqref{eq:thTwo1}; left and right multiply by $[\zeta(t)^\top, 1, u_c(t)^\top]^\top$, and consider $\alpha F'$, $\beta S'$, $\lambda T'$ as S-procedure terms by positive multipliers $\alpha$, $\beta$, and $\lambda$; this implies with the S-procedure under the change of variables:

\begin{align*}
&[\zeta^\top, 1, u_c^\top] \,E'\, [\zeta^\top, 1, u_c^\top]^\top= \dot{V}(\zeta)  \leqslant 0\\
&\text{when}\\
&[\zeta^\top, 1, u_c^\top] \, F' \, [\zeta^\top, 1, u_c^\top]^\top = V(\zeta)-1 \geqslant 0 \Leftrightarrow V(\zeta) \geqslant 1\\
&[\zeta^\top, 1, u_c^\top] \, S' \, [\zeta^\top, 1, u_c^\top]^\top  \geqslant 0 \Leftrightarrow u_c(t) \in  \mathcal{E}_u(R, \bar{u}),\\
&[\zeta^\top, 1, u_c^\top] \, T' \, [\zeta^\top, 1, u_c^\top]^\top  \geqslant 0 \Leftrightarrow \zeta(t) \in  \mathcal{E}_n(\Xi, \bar{\xi})
\end{align*}
where $\mathcal{G}$ is a lower bound of $X \Xi_p X$ coming from the change of variables:
\begin{equation}
\Pi_1^\top \Xi \Pi_1 =
\begin{bmatrix} X \Xi_p X\, &  X \Xi_p\\
			 \Xi_p X\, &  \Xi_p
\end{bmatrix}
\end{equation}
This means that the value of $V(\zeta)$ can only increase under the stated constraints, i.e.\ $V(\zeta(t_0))\leqslant 1 \Rightarrow  V(\zeta(t))\leqslant 1$ $\forall t \ge t_0$.
\end{proof}
\begin{remark}[Non-stealthy case]
In order to consider the  case of non-stealthy attacks, just remove the term $\lambda T'$ from \eqref{eq:thTwo1} in Theorem~\ref{th:two}. This term provides the extra constraint needed only if the attacker tries to keep the system within the normal operation set $X_n$; so if no normal operation set is defined then the term is not present.
\end{remark}

\begin{remark}
Theorem~\ref{th:two} is conservative, i.e. it provides sufficient but not necessary conditions, for two main reasons: 1) the Lyapunov function is constrained to be of a specific form, 2) the use of the generalized S-procedure for obtaining \eqref{eq:thTwo1} and \eqref{eq:thTwo2}, and 3) the lower bound of $X \Xi_p X$.
\end{remark}

Due to the product of $\alpha$ with $\mathcal{Q}(\kappa)$, $\lambda$ with $\mathcal{G}$, $\mathcal{H}$, and $\delta$ with $X$, the matrix inequalities \eqref{eq:thTwo1} and \eqref{eq:thTwo2} in Theorem~\ref{th:two} are not LMIs. To relax it, the invariant ellipsoidal set $\mathcal{E}_\zeta$ together with $\kappa$ will be computed for a fixed $\alpha, \lambda, \delta \ge 0$.

After having provided the sufficient conditions for synthesizing a filter $(A_f$, $B_f$, $C_f$, $D_f$) that guarantees that the projection of the invariant ellipsoidal set $\mathcal{E}_\zeta(Q)$ onto the $x_p$-hyperplane is a subset of the safe set $\mathcal{E}_s(\Psi, \bar{\psi})$, i.e. the plant state trajectories remain inside the safe set, and the $H_\infty$-norm of the transfer channel $T_{u_c \rightarrow z}(s)$ is below a $\gamma$ gain, i.e. $\lvert \lvert T_{u_c \rightarrow z}(s) \rvert \rvert_{H_\infty} \le \gamma$, we want to compute the smallest invariant ellipsoidal set on $\mathcal{R}_{\zeta}(\infty)$. This is obtained by maximizing the trace of $Q$ which is similar to minimize the trace of $X$ under the change of variables. Thus, we want to solve the following optimization problem under LMI constraints, \textbf{OP}.\\

\noindent \textbf{OP: \quad Filter synthesis}
\vspace{-1mm}
  \begin{align*}
    \underset{X, Y, \hat{A}_f, \hat{B}_f, \hat{C}_f, \hat{D}_f, \beta}{\text{minimize}} & \qquad \textrm{trace}(X)\\
  \text{subject to}
& \qquad \eqref{eq:thTwo1}, \eqref{eq:thTwo2}, \eqref{eq:thTwo3}, \eqref{eq:thTwo4}
  \end{align*}

\section{Example} \label{sec:example}
In this section, we propose to apply the proposed method on a dynamical system having two actuators. Consider stealthy actuator attacks that aim to damage the actuators by injecting malicious control signals from either the controller or the communication network. First, we analyze the effect of stealthy actuator attacks using the Lemma~\ref{lem:invar}. Then, we synthesize a filter to prevent such attacks by following the proposed method. We use the solver Mosek with the Yalmip toolbox on Matlab to solve the optimization problems.

\subsection{Description of the system}
Consider the dynamical system in \eqref{eq:system_example} with $x_p = [x_{p1}, x_{p2}, x_{p3}]^\top$ ($n_p = 3$) and $u_p = [u_{p1}, u_{p2}]^\top$ ($m = 2$) where $x_{p1}$ is the process state and $x_{p2}$, $x_{p3}$ are the states of two actuators.
\begin{equation}\label{eq:system_example}
A_p =
\begin{bmatrix}
-10\, & 10\, & 10\\
0\, & -150\, & 0\\
0\, & 0\, & -150
\end{bmatrix}
\,
,
\,
B_p =
\begin{bmatrix}
0\, & 0\\
100\, & 0\\
0\, & 100
\end{bmatrix}
\end{equation}
Consider the input set $\mathcal{E}_u(R,\bar{u})$ defined for $R = \mbox{diag(}0.25, 0.25\mbox{)}$, $\bar{u} = [0, 0]^\top$, the safe set $\mathcal{E}_s(\Psi, \bar{\psi})$ defined for $\Psi = \mbox{diag(}0.001, 0.0156, 0.0156\mbox{)}$, $\bar{\psi} = [0, 0, 0]^\top$, and the normal operation set $\mathcal{E}_n(\Xi, \bar{\xi})$, defined for $\Xi_p = \mbox{diag(}0.01, 0.001, 0.001\mbox{)}$, $\bar{\psi}_p = [0, 0 ,0]^\top$. This means that the process state $x_{p1}$ is constrained by the normal operation set $\mathcal{E}_n$, whereas the actuator states are not. The safe set defines a safe zone where the actuator states must remain in order to not damage the actuators, that is the actuator states $x_{p2}$, $x_{p3}$ are constrained by the safe set $\mathcal{E}_s$, whereas the process state is not.

\subsection{Attack analysis}
Consider first the problem of analyzing the effect of stealthy actuator attacks to the plant. This problem consists in computing the smallest invariant ellipsoidal set by solving an optimization problem maximizing the trace of Q under constraints \eqref{eq:thun1}, \eqref{eq:thun2} for a fixed $\alpha \ge 0$. Consider that no filter is placed, i.e. $u_p(t) = u_c(t)$ (see Figure~\ref{fig:sys_with_filter}). This can be set by letting $n_f = 0$, that is the filter state $x_f$ can be removed from the extended system in \eqref{eq:extended_system}, so $A_f$, $B_f$, and $C_f$ do not exist, and $D_f = I_m$.
For $\alpha = 0.5$, the result is drawn in Figure~\ref{fig:analysis} (left-hand side) where the projection of the smallest invariant ellipsoidal set $\mathcal{E}_\zeta$ onto the $x_p$-hyperplane is the ellipsoid filled in green, the safe set $\mathcal{E}_s$ is the ellipsoid filled in blue, and the normal operation set $\mathcal{E}_n$ is the ellipsoid filled in magenta.

\begin{figure*}[h]
    \centering
    \begin{subfigure}[b]{\textwidth} 
    	\includegraphics[width=0.5\textwidth]{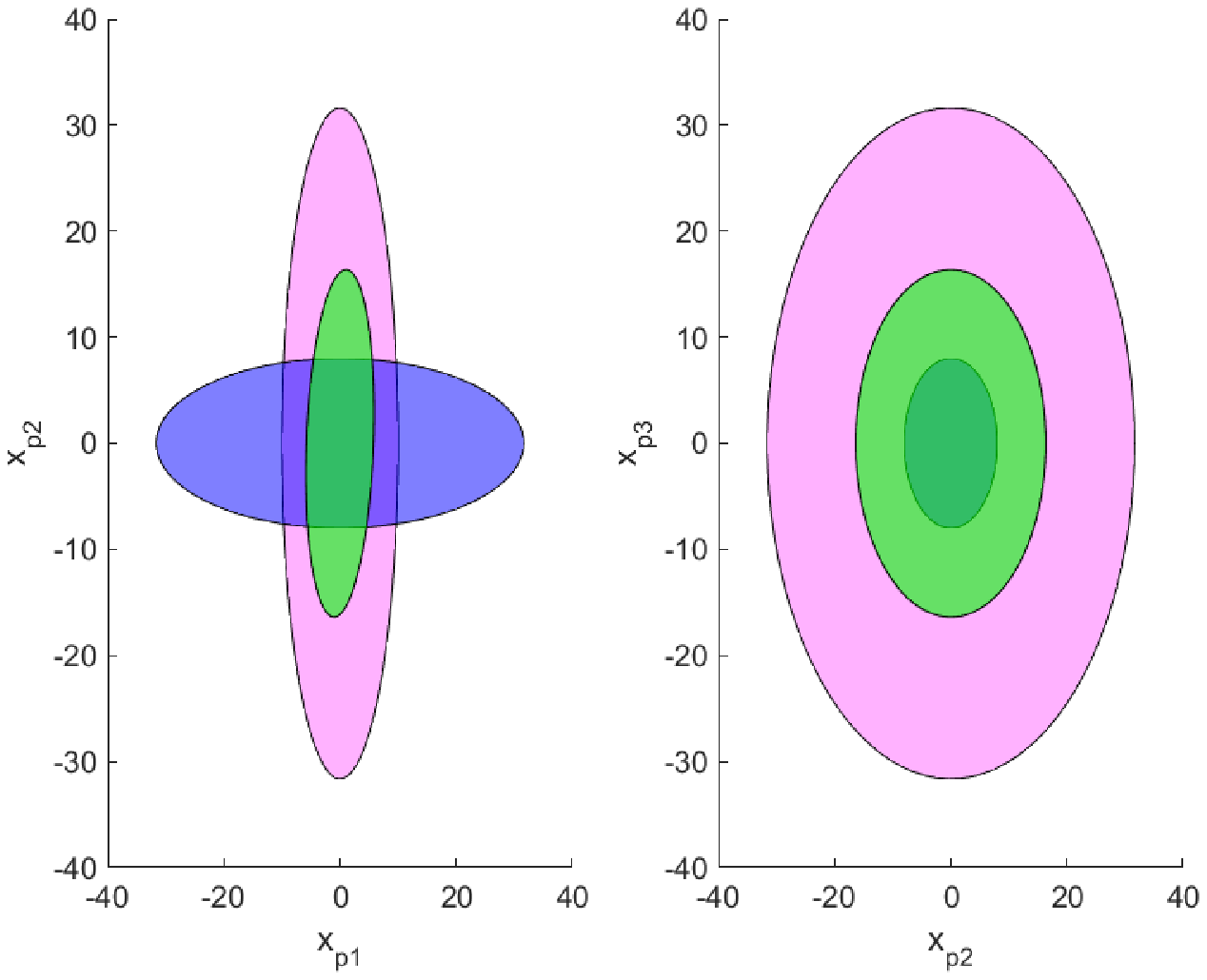}
    	\includegraphics[width=0.5\textwidth]{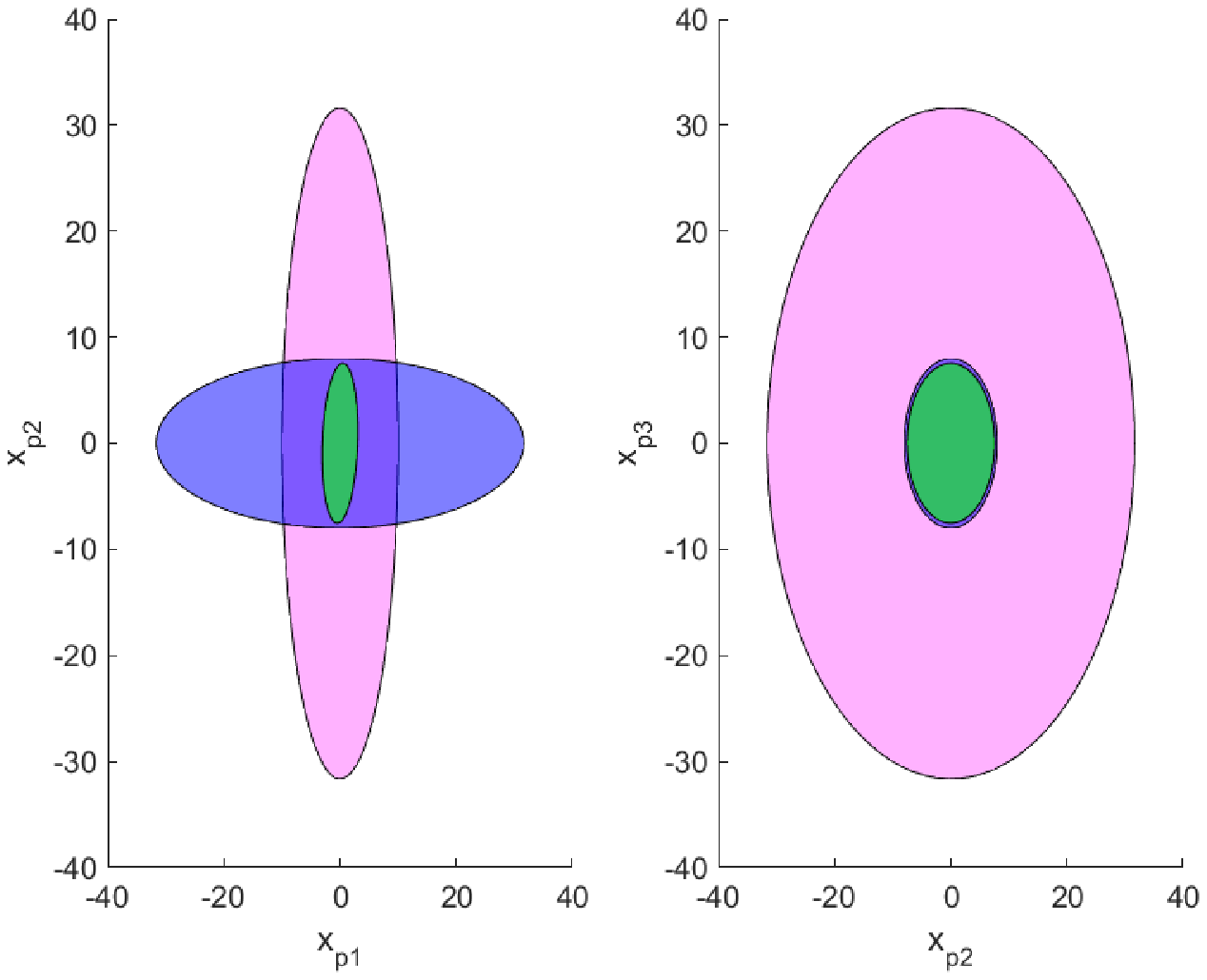}
    \end{subfigure}
    \caption{Projection of the invariant ellipsoidal set $\mathcal{E}_\zeta(Q)$ onto the $x_p$-hyperplane (green fill), normal operation set $\mathcal{E}_n(\Xi,\bar{\xi})$ (magenta fill), safe set $\mathcal{E}_s(\Psi, \bar{\psi})$ (blue fill) - Attack analysis (left-hand side): stealthy actuator attacks are feasible as $\mathcal{E}_{x_p} \nsubseteq \mathcal{E}_s$ - Attack prevention (right-hand side): stealthy actuator attacks are prevented with the computed filter as $\mathcal{E}_{x_p} \subseteq \mathcal{E}_s$}\label{fig:analysis}
\end{figure*}

We can easily observe that the projection of the invariant ellipsoidal set onto the $x_p$-hyperplane is not a subset of the safe set, that is stealthy actuator attacks are feasible.

\subsection{Attack prevention: synthesis of the filter}
Consider now the synthesis problem of the filter. This problem consists in finding the filter matrices ($A_f$, $B_f$, $C_f$, $D_f$) such that the projection of the invariant ellipsoidal set $\mathcal{E}_\zeta(Q)$ onto the $x_p$-hyperplane, i.e. $\mathcal{E}_{x_p}$, is a subset of the safe set $\mathcal{E}_s(\Psi,\bar{\psi})$ and the $H_\infty$-norm of the transfer matrix $T_{u_c \rightarrow z}(s)$ is below a given level $\gamma$, i.e. $\lvert \lvert T_{u_c \rightarrow z}(s) \rvert \rvert_{H_\infty} \le \gamma$. The synthesis problem is answered by solving the optimization problem \textbf{OP} for fixed $\alpha, \lambda, \delta, \gamma, \epsilon \ge 0$

Consider that the filter we want to synthesize actuates on the entire control signals $u_c(t)$ transmitted to the plant, i.e. $\Gamma_f = I_m$, $\Gamma_c = \mathbf{0}$. For $\alpha = 1$, $\beta = 0.4999$, $\lambda = 0.5$, $\delta = 0.9$, $\gamma = 0.61$, and $\epsilon = 10^{-8}$, the optimization problem is solved. The result is drawn in Figure~\ref{fig:analysis} (right-hand side) where the projection of the smallest invariant ellipsoidal set $\mathcal{E}_\zeta$ onto the $x_p$-hyperplane is the ellipsoid filled in green, the safe set $\mathcal{E}_s$ is the ellipsoid filled in blue, and the normal operation set $\mathcal{E}_n$ is the ellipsoid filled in magenta.

\noindent The obtained filter matrices are given as follows.
\begin{align}\label{eq:filter_example}
A_f &=
\begin{bmatrix}
-12.75\, & 22.55\, & 22.55\\
8.35\, & -151.39\, & 1.31\\
8.35\, & 1.31\, & -151.39
\end{bmatrix}
\,
, \nonumber
\\ \nonumber
B_f &=
\begin{bmatrix}
-549.65\, & -549.65\\
-647.31\, & -35.94\\
-35.94\, & -647.31
\end{bmatrix}
\,
,
\\ \nonumber
C_f &=
\begin{bmatrix}
-1 \times 10^{-4}\, & 1.7 \times 10^{-3}\, & 0\\
-1 \times 10^{-4}\, & 0\, & 1.7 \times 10^{-3}
\end{bmatrix}
\,
,
\\
D_f &=
\begin{bmatrix}
0.46\, & 2 \times 10^{-4}\\
2 \times 10^{-4}\, & 0.46
\end{bmatrix}
\end{align}

We can observe that the projection of the invariant ellipsoidal set onto the $x_p$-hyperplane is now a subset of the safe set by filtering the control signals $u_c(t)$ with the obtained filter in \eqref{eq:filter_example}, meaning that stealthy actuator attacks are not feasible.

In Figure~\ref{fig:bode_fil_eps} (left-hand side), the Bode diagrams of the system without the filter, i.e. the plant only, (blue) and with the filter (red), i.e. the plant in series with the filter, are drawn. Firstly, we can observe that the filter only changes the system dynamics slightly, which is due to the imposed constraint $\lvert \lvert T_{u_c \rightarrow z}(s) \rvert \rvert_{H_\infty} \le \gamma$. Secondly, we can see that by placing the filter there are now a relationship between the actuator state $x_{p3}(t)$ and the control input $u_{c1}(t)$ and the actuator state $x_{p2}$ and the control input $u_{c2}(t)$.

Figure~\ref{fig:bode_fil_eps} (right-hand side) shows the Bode plot of the filter alone. It can be noticed that the action is not the same at all frequencies, which confirms the idea mentioned in the Introduction, i.e.\ that this approach by dynamical constraints (filter) is less restrictive than the hard static bound proposed in  \cite{esc_iet}. It is interesting to point out that the filter works mainly on the coupling between the actuators (i.e.\ $u_1$ vs  $u_2$), and with a very small gain it manages to reduce the reachable set within the safe set in its entirety (right-hand side of Figure~\ref{fig:analysis}).

\begin{figure*}[h]
    \centering
    \begin{subfigure}[b]{\textwidth} 
    	\includegraphics[width=0.5\textwidth, trim = 0.9cm 0cm 1.1cm 0cm, clip]{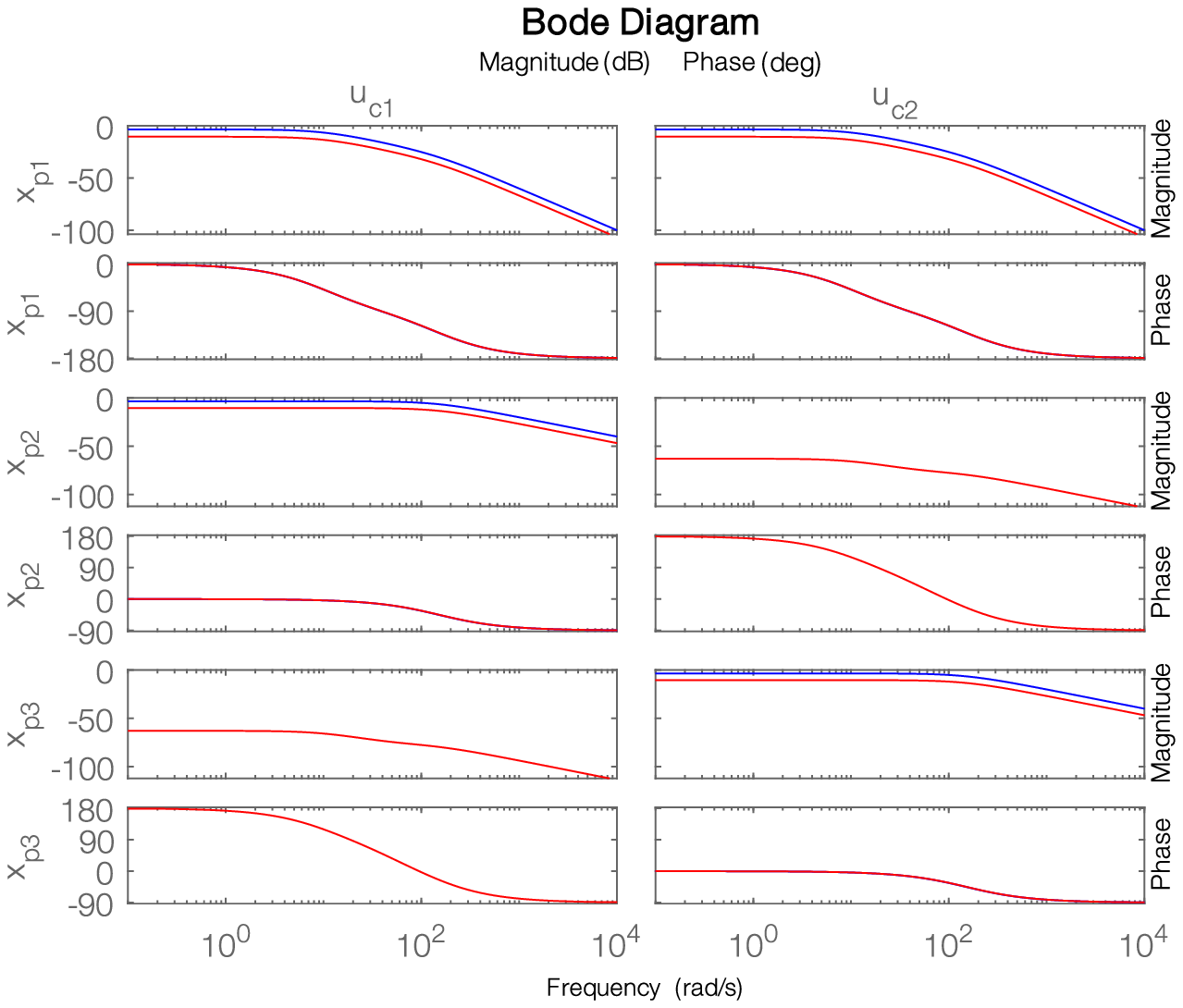}
        \includegraphics[width=0.5\textwidth, trim = 1.7cm 0cm 1.1cm 0cm, clip]{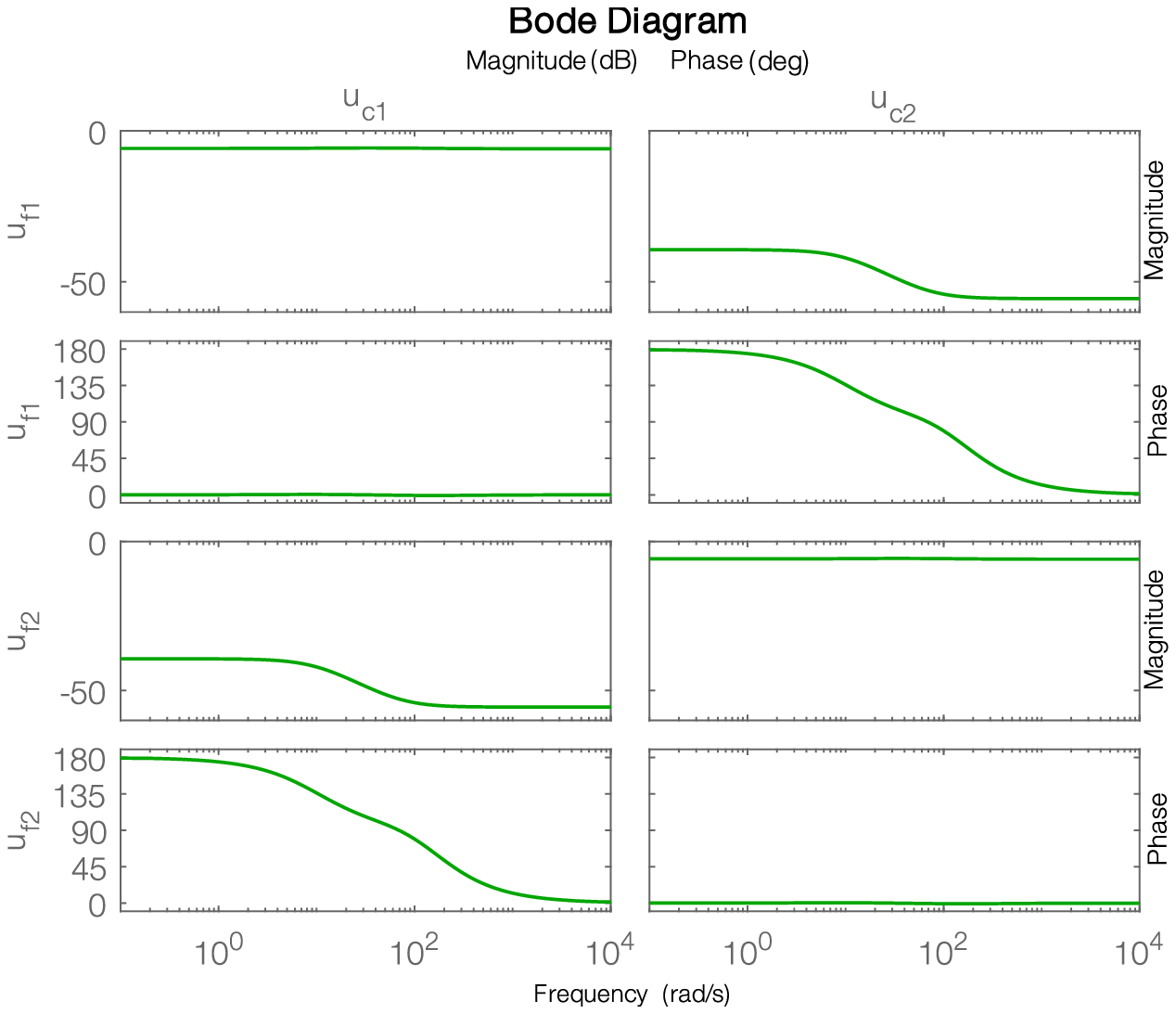}
    \end{subfigure}

    \caption{(Left-hand side) Bode diagrams of the system without the filter (blue) and with the filter (red) - (Right-hand side) Bode diagrams of the filter alone}\label{fig:bode_fil_eps}
\end{figure*}
%
%
%

\section{Conclusion}
In this manuscript, we have proposed a set-theoretical method to synthesize filters working on the controller output received from a communication network to prevent actuator injection attacks. The focus of this work has been on the introduction of this novel idea, with some preliminary results and an academic example; on the other hand we can envisage several extensions as topics of future research. First, future works will address the controllability of the plant where the filter is placed. The second issue to be addressed as a perspective is the co-design of a controller together with the filter to satisfy a trade-off between safety and control objectives. At last, so far we have considered filters without state feedback from sensors, but we could extend them by considering that some states are securely monitored and sent back to the filter, allowing for less restrictive filtering.

\bibliographystyle{IEEEtran}
\bibliography{main}

\end{document}